\documentclass[12pt,twoside,english]{article}
\usepackage{mathptmx}

\usepackage[T1]{fontenc}
\usepackage[latin9]{inputenc}
\usepackage{geometry}
\usepackage{babel}
\usepackage{array}
\usepackage{multirow}
\usepackage{amsmath}
\usepackage{amsthm}
\usepackage{setspace}
\usepackage[unicode=true,pdfusetitle,
 bookmarks=true,bookmarksnumbered=false,bookmarksopen=false,
 breaklinks=false,pdfborder={0 0 1},backref=false,colorlinks=false]
 {hyperref}

\makeatletter

\providecommand{\tabularnewline}{\\}

\theoremstyle{plain}
\newtheorem{thm}{\protect\theoremname}
\theoremstyle{plain}
\newtheorem{assumption}{\protect\assumptionname}
\theoremstyle{remark}
\newtheorem{rem}{\protect\remarkname}

\usepackage{setspace}
\usepackage{footmisc}

\usepackage{etoolbox}
\patchcmd{\thanks}{#1}{\protect\doublespacing\normalsize#1}{}{}

\renewcommand{\footnotesize}{\normalsize}

\makeatother

\providecommand{\assumptionname}{Assumption}
\providecommand{\remarkname}{Remark}
\providecommand{\theoremname}{Theorem}

\begin{document}
\title{A dynamic ordered logit model with fixed effects}
\author{Chris Muris, Pedro Raposo, and Sotiris Vandoros\thanks{Muris: Department of Economics, McMaster University. Contact: muerisc@mcmaster.ca.
Raposo: Católica Lisbon School of Business and Economics. Contact:
pedro.raposo@ucp.pt. Vandoros: King\textquoteright s College London
and Harvard University. Contact: s.vandoros@kcl.ac.uk. The data were
made available by Eurostat (Contract RPP 132-2018-EU-SILC). We are
grateful to Irene Botosaru and Krishna Pendakur for very helpful discussions.}}
\date{August 4, 2020}
\maketitle
\begin{abstract}
We study a fixed-$T$ panel data logit model for ordered outcomes
that accommodates fixed effects and state dependence. We provide identification
results for the autoregressive parameter, regression coefficients,
and the threshold parameters in this model. Our results require only
four observations on the outcome variable. We provide conditions under
which a composite conditional maximum likelihood estimator is consistent
and asymptotically normal. We use our estimator to explore the determinants
of self-reported health in a panel of European countries over the
period 2003-2016. We find that: (i) the autoregressive parameter is
positive and analogous to a linear AR(1) coefficient of about 0.25,
indicating persistence in health status; (ii) the association between
income and health becomes insignificant once we control for unobserved
heterogeneity and persistence.
\end{abstract}

\section{Introduction}

Certain individual-level conditions may tend to persist over time,
in the sense that a condition has a memory of a previous period\textquoteright s
state, or may involve an element of adaptation. Furthermore, the way
individuals experience the same condition may vary, and they may also
have a different understanding of how this is measured. A common example
that fulfils these characteristics, and which is used extensively
in the literature, is self-reported health status. Health status often
depends on its value in the previous period, as health conditions
may persist over time, given that recovery can take long, and that
an illness may even have permanent effects. For example, Table \ref{tab:transition-uk}
presents a \textit{\emph{transition matrix}}\textit{ }for self-reported
health status in the United Kingdom for the period 2003-2016.\footnote{More information about the data and source is in Section \ref{sec:health},
where we analyze this data using the methodology proposed in this
paper.}
\begin{table}
\begin{centering}
\begin{tabular}{ccccccc}
\hline 
 &  & \multicolumn{5}{c}{$P\left(\left.Y_{i,t+1}=y'\right|Y_{i,t}=y\right)$}\tabularnewline
 & $y/y'$ & 1 & 2 & 3 & 4 & 5\tabularnewline
\hline 
$P\left(Y_{i,t}=y\right)$ & 1 & 36.48 & 43.40 & 13.84 & 5.03 & 1.26\tabularnewline
 & 2 & 10.23 & 44.44 & 35.38 & 8.77 & 1.17\tabularnewline
 & 3 & 0.88 & 10.23 & 52.18 & 30.69 & 6.02\tabularnewline
 & 4 & 0.15 & 1.08 & 14.87 & 59.74 & 24.16\tabularnewline
 & 5 & 0.08 & 0.29 & 4.18 & 33.38 & 62.07\tabularnewline
\hline 
\end{tabular}
\par\end{centering}
\caption{Current and future self-reported health, United Kingdom.}
\label{tab:transition-uk}
\end{table}
 For individuals that report a value of current health in a given
year (rows, on a 5-point scale with 5 being the highest), it shows
the relative proportion of those that report a certain value in the
subsequent year (columns). A striking feature of this transition matrix
is that a lot of mass is on or near the main diagonal. This feature
is found across all countries in our analysis. In other words, self-reported
health status is persistent: individuals tend to stay in the same
level of health.

There are at least two explanations for this observed persistence
(Heckman, 1981; Honoré and Kyriazidou, 2000): unobserved heterogeneity
and state dependence. Consider first \emph{unobserved heterogeneity.
}It\emph{ }refers to unobservable characteristics that affect the
propensity to report higher health status. Unobserved heterogeneity
is important in the literature on health status, because self-reported
health has been used extensively in the literature as a measure of
health outcomes (see for example Bound and Waidmann, 1992; Banerjee
et al., 2004; Gravelle and Sutton, 2009; McInerney and Mellor, 2012).
It is often viewed as a limitation that self-reported measures are
subjective. For example, reporting one's own health may depend on
cultural factors (Jylhä et al., 1998; Baron-Epel et al., 2005; Jürges,
2007), and people may have a different understanding of reference
points for health (Groot, 2000; Sen, 2002). Previous studies have
used vignettes to address cross-country differences in reporting of
health and disability (King et al., 2004; Salomon et al., 2004; Kapteyn
et al., 2007). However, the issue with unobserved heterogeneity across
individuals remains. As a result, it is important to take into account
the role of unobserved heterogeneity when analyzing self-reported
health data. The appropriate econometric approach to this is to allow
for fixed effects.\footnote{Studies have long debated the accuracy and reliability of subjective
measures of health, such as self-reported health status (see for example
Butler et al., 1987; Lindeboom and van Doorslaer, 2004; Johnston et
al., 2009). As an alternative response to these concerns, a number
of objective measures of health have been included in household surveys.
These include blood pressure, BMI, the number of medicines taken (Health
Survey England, 2019), the number of sick days off work, the number
of days hospitalised (BHPS, 2019) etc. Some household surveys ask
respondents to perform a task such as walking across the room or buttoning
a shirt to capture any limitations (SHARE, 2019). Indexes such as
the EQ-5D index are being used to cover different types of conditions
and merge them into a single measure. The Euro-D scale measures mental
health, and the CASP-12 index captures quality of life. However, these
objective measures are often very specific to particular diseases,
and even when creating a relevant index it may be impossible to include
and accurately reflect all conditions. As such, while objective measures
may accurately capture some health conditions, they have serious limitations
in capturing the overall picture of one\textquoteright s health.} 

A number of studies have found a positive association between income
and health (Carrieri and Jones, 2017; Ettner, 1996; Frijters et al.,
2005; Mackenbach et al., 2005), but empirical evidence of a strong
effect is sometimes limited (Larrimore, 2011; Gunasekara, 2011; Johnston
et al., 2009). Nevertheless, the literature on the impact of economic
downturns (which mean reduced income) has previously demonstrated
positive effects of unemployment on health (Ruhm, 2000; Ruhm and Black,
2002), and more recently, no effect (Ruhm, 2015). What also appears
to matter, at least in terms of happiness, apart from absolute income,
is also relative income, i.e. how one\textquoteright s income compares
to that of those around them (Frijters et al., 2008). The relationship
between income and health is endogenous and complex, and both can
be correlated with other factors, that are not always measured and
included in empirical models. For example, Gunasekara et al. (2011)
found that when controlling for unmeasured confounders, the association
between the two becomes weaker. With regards to self-reported hypertension
in particular, Johnston et al., (2009) found no link to income --
something that did change when using objective measures. In our paper,
using models that do not control for individual unobserved heterogeneity
yields a positive and statistically significant coefficient (Table
3). However, this becomes insignificant when using fixed effects,
suggesting that there are other factors that potentially drive the
association between the two variables.

Consider now the second explanation for the observed persistence in
health outcomes: \emph{state dependence}. It refers to the possibility
that past self-reported health status may be related to current self-reported
health status even after conditioning on unobserved heterogeneity.
State dependence arises if actual (as opposed to self-reported) health
shocks are persistent, in the sense that a shock on health can have
a long-lasting effect (a typical example is injury leading to disability).
Contoyannis et al. (2004), using a random effects approach, found
evidence for such persistence in respondents of the British Household
Panel survey.

State dependence in self-reported health can also arise due to adaptation:
self-reported health status may change over time for a person whose
actual health has not changed. People tend to adapt to good or bad
developments in life. According to the Global Adaptive Utility Model,
individuals reallocate weights on various domains of life in order
to maintain their previous level of utility (Bradford and Dolan, 2010).
Similarly, the AREA model developed by Wilson and Gilbert (2008),
suggests that attention is focused on a change, followed by reaction,
explanation, and, finally, adaptation. This also applies to health,
as health status tends to improve even when individuals' health has
actually not experienced any objective change (Daltroy et al., 1999;
Damschroder et al., 2005), and time since diagnosis is positively
associated with self-reported health (Cubí-Mollá et al., 2017). Whether
persistence or adaptation, or both, characterise a variable, this
calls for a dynamic element in a model. 

Overall, the challenges with studying self-reported health status
is that (a) people with the same actual health status might be reporting
different health levels; and (b) health shocks can have a lasting
effect. Against this background, we propose and analyze a panel data
ordered logit model that includes both fixed effects and a lagged
dependent variable. This allows a researcher faced with panel data
and an ordinal outcome variable to disentangle unobserved heterogeneity
from state dependence, and to quantify state dependence. Thus, we
address the limitations of using self-reported health as a proxy for
individuals' health. Our contribution is important for studies using
subjective health measures as it can help correct biases that naturally
occur when using this type of measure.\footnote{For example, happiness is perceived and reported differently across
individuals and people adapt to things that make them happy (Layard,
2006), while shocks on happiness can have a scarring effect on next
periods (Clark et al., 2001).}

Specifically, we study the \textit{dynamic ordered logit model with
fixed effects:}
\begin{align}
Y_{i,t}^{*} & =\alpha_{i}+X_{i,t}\beta+\rho1\left\{ Y_{i,t-1}\geq k\right\} -U_{i,t},\,t=1,2,3,\label{eq:latent_variable}\\
Y_{i,t} & =\begin{cases}
1 & \text{if }Y_{i,t}^{*}<\gamma_{2},\\
2 & \text{if }\gamma_{2}\leq Y_{i,t}^{*}<\gamma_{3},\\
\vdots\\
J & \text{if }Y_{i,t}^{*}\geq\gamma_{J},
\end{cases}\label{eq:thresholds}\\
\left.U_{i,t}\right| & \left(\alpha_{i},X_{i},Y_{i,<t}\right)\sim LOG(0,1),\,t=1,2,3,\label{eq:strict_exogeneity}
\end{align}
where $2\leq k\leq J$ is a fixed and known cutoff for the lagged
dependent variable. The person-specific parameter $\alpha_{i}$ captures
unobserved heterogeneity, which we allow to be correlated with the
other quantities in the model in an unrestricted way (fixed effects).
The time-varying covariates $X_{i,t}$ are collected across time periods
in $X_{i}=\left(X_{i,1},X_{i,2},X_{i,3}\right)$, and the lagged dependent
variables for period $t$ are collected in $Y_{i,<t}=\left(Y_{i,0},\cdots,Y_{i,t-1}\right)$.
The autoregressive parameter $\rho$ is the regression coefficient
on the lagged dependent variable $1\left\{ Y_{i,t-1}\geq k\right\} $;
$\beta$ is the regression coefficient on the contemporaneous covariates;
and the threshold parameters $\gamma_{j}$ map the underlying latent
variable $Y_{i,t}^{*}$ into the observed ordered outcome $Y_{i,t}$.
Equation (\ref{eq:strict_exogeneity}) restricts the error terms $U_{i,t}$
to be i.i.d. logistic, and is a strict exogeneity assumption on the
regressors and past outcomes.\footnote{The dynamics in our model are restricted to depend on $Y_{i,t}$ through
$1\left\{ Y_{i,t-1}\geq k\right\} $ only. An alternative model for
which we can identify some features is one that is linear in its history,
i.e. $Y_{i,t}^{*}=\alpha_{i}+X_{i,t}\beta+\rho Y_{i,t-1}-U_{i,t}.$
We were unable to use our approach to obtain identification in the
more general model with $Y_{i,t}^{*}=\alpha_{i}+X_{i,t}\beta+\sum_{j=2}^{J}\rho_{j}1\left\{ Y_{i,t-1}=j\right\} -U_{i,t}$.} 

This model combines a number of noteworthy features. First, it is
a model for discrete ordered outcomes, and therefore a \textit{nonlinear}
model. Second, it is \textit{dynamic}, in the sense that the current
outcome depends directly on the outcome in the previous period. This
feature, called \textit{state dependence}, is governed by the autoregressive
parameter $\rho$. Third, it allows for \textit{unobserved heterogeneity}
in an unrestricted way, i.e. it is a \emph{fixed effects }model. Fourth,
the model is only specified for a \textit{small number of time periods},
$T=3$. Period 0 is unmodelled, but an observation on the outcome
variable in time 0 is required for identification.

We believe that we are the first to provide identification and estimation
results for all common parameters in a dynamic ordered logit model
with fixed effects and a fixed number of time periods. Using four
time periods of data on the ordinal outcome variable, we identify
the autoregressive coefficients on the lagged dependent variable,
and the regression coefficients on the exogenous regressors. We also
identify the threshold parameters, which makes it possible to interpret
the magnitude of the estimated coefficients. This distinguishes the
ordered choice model from the dynamic binary choice model with fixed
effects, where such an interpretation is not available. Our identification
result suggest a composite conditional maximum likelihood estimator
for the parameters in our model. We establish conditions under which
that estimator is consistent and asymptotically normal. 

We use our estimator to investigate the determinants of self-reported
health, focusing on the link between income and health in a panel
of European countries over the period 2003-2016. We obtain two main
findings. First, even after controlling for unobserved heterogeneity,
persistence plays a positive and significant role in one's self-reported
health. In other words, one's health is dependent on the health in
the previous period, which is a reasonable thing to expect, as health
problems may expand over a number of periods, or become permanent.
Quantitatively, we estimate a persistence parameter that is analogous
to an autoregressive parameter of about 0.25 in a linear AR(1) model.
Second, we find that, when controlling for unobserved heterogeneity,
the link between income and health becomes statistically insignificant,
suggesting that other factors might explain the association between
the two. This is in line with studies that have found a smaller or
insignificant association when using fixed effects (Gunasekara, 2011;
Larrimore, 2011).

\section{Related literature in econometrics}

We believe that our paper is the first to provide identification and
estimation results for a panel data model with (i) ordered outcomes;
(ii) a lagged dependent variable; (iii) fixed effects; and (iv) a
fixed number of time periods. Our econometric contribution is related
to several strands of literature, each of which features a subset
of these features.

Most closely related to our paper is the literature on binary and
multinomial choice models with fixed effects and lagged dependent
variables, which features all but (i). The seminal work by Honoré
and Kyriazidou (2000) builds on Cox (1958) and Chamberlain (1985)
to estimate the parameters in dynamic binary choice logit model with
fixed effects and time-varying regressors. Hahn (2001) discusses the
information bound for a special case of their model. Honoré and Kyriazidou
(2019) discuss identification of some closely related models. Honoré
and Weidner (2020) construct moment conditions that shed light on
identification in this and related models, and provide a $\sqrt{n}$-consistent
estimator. Honoré and Tamer (2006), Aristodemou (2020) and Khan et
al. (2020) obtain results for models that do not have logistic errors.
For the static multinomial model, Chamberlain (1980) studies the logit
case; Shi et al. (2008) provides results for the general static; and
Magnac (2000) studies the dynamic version. We supplement these results
by showing that, in an \emph{ordered} choice model, the thresholds
in the latent variable model can be identified along with the regression
coefficients and the autoregressive parameter. This allows for a quantitative
interpretation of true state dependence. Such an interpretation is
not available in the binary and multinomial choice models.

The literature on static ordered logit models with fixed effects features
all but (ii). This model was analyzed by Das and van Soest (1999),
Baetschmann et al. (2015), and Muris (2017). Our result differs from
the results in those papers, because we provide results for a \emph{dynamic}
version of the ordered logit model.

The literature on random effects dynamic ordered choice models features
all but (iii). Random effects dynamic ordered choice models have been
studied and applied extensively (Contoyannis et al., 2004; Albarran
et al., 2019). Such approaches require strong restrictions on the
relationship between the unobserved heterogeneity and the exogeneous
variables in the model. Such restrictions are usually unappealing
to economists, as evidenced by the fact that they are rarely used
in linear models. Our approach does not impose random effects restrictions
and is the first to provide a fixed effects approach for dynamic ordered
choice models. 

Note that our approach is fixed-$T$ consistent. The difficulty of
allowing for fixed effects is alleviated when one can assume that
$T\to\infty$, referred to as ``large-$T$''. Large-T fixed effects
dynamic ordered choice models have been studied by Carro and Traferri
(2014) and Fernández-Val et al. (2017), see also Carro (2007) for
the binary outcome case. In the large-$T$ case, one can use techniques
that correct for the bias that comes from including fixed effects
in the nonlinear panel model. This approach does not feature (iv).
These techniques are not appropriate for our empirical application,
which is a rotating panel with $T=4$.

One limitation of our approach is that we restrict the way in which
the lagged dependent variable enters the model. The random effects
and large-$T$ approach can accommodate a richer dynamic specification.
We leave for future work whether such an extension is possible with
a fixed-$T$ fixed-effects approach.

\section{Identification\label{sec:Identification}}

We normalize $\gamma_{k}=0$, where $k$ is as in equation (\ref{eq:latent_variable}).
This scale normalization is without loss of generality because the
scale of $\alpha_{i}$ is unrestricted. Our model implies that the
binary variable $D_{i,t}(k)=1\left\{ Y_{i,t}\geq k\right\} $ follows
the dynamic binary choice logit model in Honoré and Kyriazidou (2000),
HK hereafter. Specifically, equation (3) in HK applies to the transformed
model
\[
D_{i,t}(k)=1\left\{ X_{i,t}\beta+\rho D_{i,t-1}\left(k\right)+\alpha_{i}-U_{i,t}\geq0\right\} ,
\]
i.e. the transformed model follows a dynamic binary choice logit model
with fixed effects. The implied conditional probabilities relevant
for our analysis are 
\begin{equation}
P\left(\left.D_{i,0}\left(k\right)=1\right|X_{i},\alpha_{i}\right)\equiv p_{0}\left(X_{i},\alpha_{i}\right),\label{eq:model_probabilities_0}
\end{equation}
and, for $t=1,2,3$,

\begin{align}
P\left(\left.D_{i,t}\left(k\right)=1\right|X_{i},\alpha_{i},D_{i,<t}\left(k\right)\right) & =\frac{\exp\left(\alpha_{i}+X_{i,t}\beta+\rho D_{i,t-1}\left(k\right)\right)}{1+\exp\left(\alpha_{i}+X_{i,t}\beta+\rho D_{i,t-1}\left(k\right)\right)},\label{eq:model_probabilities_t}
\end{align}
where we have let $D_{i,<t}\left(k\right)=\left(D_{i,0}\left(k\right),\cdots,D_{i,t-1}\left(k\right)\right)$.
HK provide conditions that guarantee identification of $\beta$ and
$\rho$ by constructing a conditional probability that features $\left(\beta,\rho\right)$
but that is free of $\alpha_{i}$.

If $Y_{i,t}$ has at least three points of support, there is information
in $Y_{it}$ beyond $D_{it}\left(k\right)$. In the remainder of this
section, we show that this information can be used to identify the
threshold parameters 
\[
\gamma\equiv\left(\gamma_{2},\gamma_{3},\cdots,\gamma_{k-1},\gamma_{k+1},\cdots,\gamma_{J}\right).
\]
This leads to an interpretation of the magnitude\textbf{ }of $\left(\beta,\rho\right)$
that is not available for the dynamic binary choice model. Muris (2017,
Section III.C) discusses this for the static panel data ordered choice
models ($\rho=0$).

We now construct a conditional probability that features $\left(\beta,\rho,\gamma\right)$
but not the incidental parameters $\alpha_{i}$. To this end, extend
the definition 
\[
D_{i,t}\left(j\right)=1\left\{ Y_{i,t}\geq j\right\} ,\,2\leq j\leq J,
\]
to thresholds $j\neq k$, and abbreviate $D_{i,t}\equiv D_{i,t}\left(k\right)$.
Define the events $\left(A_{j,l},B_{j,l},C_{j,l}\right)$, with $2\leq j\leq k\leq l\leq J$,\footnote{Choosing $j\leq k$ guarantees that when $D_{i,2}\left(j\right)=0,$
the lagged dependent variable in period 3 is 0. The opposite is true
for $l\geq k$ and $D_{i,2}\left(l\right)=1$. There would be no gain
from considering a threshold different from $k$ in the first period.
Using $k$ as the threshold in the second period is the only way to
cancel out the threshold parameters from period 1. Given that we are
using the subpopulation $X_{i,2}=X_{i,3}$, the fact that $j,l$ are
used alternately in periods 2 and 3 does not create additional difficulties.} as follows:
\begin{align*}
A_{j,l} & =\left\{ D_{i,0}=d_{0},D_{i,1}=0,D_{i,2}\left(l\right)=1,D_{i,3}\left(j\right)=d_{3}\right\} ,\\
B_{j,l} & =\left\{ D_{i,0}=d_{0},D_{i,1}=1,D_{i,2}\left(j\right)=0,D_{i,3}\left(l\right)=d_{3}\right\} ,\\
C_{j,l} & =A_{j,l}\cup B_{j,l}.
\end{align*}
For $d_{0}=d_{3}=0$, the event $A_{j,l}$ corresponds to moving up
in the middle periods $t=1,2$, starting below $k$ to moving up to
at least $l\geq k$. The event $B_{j,l}$ corresponds to moving down
in the middle periods, starting from at least $k$ and moving below
$j\leq k$. 

If $j=k=l$, the event $C_{k,k}$ corresponds to switchers (observations
with $D_{i1}+D_{i2}=1$), as in HK. In the ordered model, it is possible
to vary the cutoffs in the periods $t=2,3$ if the dependent variable
has more than two points if support. Varying the cutoffs over time
is what distinguishes our conditioning event from that in HK. It is
what allows us to identify the threshold parameters.

The following sufficiency result shows that different choices of $\left(j,l\right)$
reveal different combinations of thresholds in certain conditional
probabilities that do not depend on the incidental parameters $\alpha_{i}$.
In what follows, the logistic function is denoted by $\Lambda\left(u\right)=\exp\left(u\right)/\left(1+\exp\left(u\right)\right)$,
and the change in the regressors from period 1 to 2 by $\Delta X_{i}=X_{i2}-X_{i1}$.
\begin{thm}[Sufficiency]
\label{thm:sufficiency}For the dynamic ordered logit model with
fixed effects, for any $\left(j,l\right)$ such that $2\leq j\leq k\leq l\leq J$,
and for any $d_{0},d_{3}\in\left\{ 0,1\right\} ,$
\begin{align}
P\left(\left.A_{j,l}\right|X_{i},C_{j,l},X_{i,2}=X_{i,3}\right) & =1-\Lambda\left(\Delta X_{i}\beta+\rho\left(d_{0}-d_{3}\right)+\left(1-d_{3}\right)\gamma_{l}+d_{3}\gamma_{j}\right)\label{eq:sufficiency_A}\\
P\left(\left.B_{j,l}\right|X_{i},C_{j,l},X_{i,2}=X_{i,3}\right) & =\Lambda\left(\Delta X_{i}\beta+\rho\left(d_{0}-d_{3}\right)+\left(1-d_{3}\right)\gamma_{l}+d_{3}\gamma_{j}\right).\label{eq:sufficiency_B}
\end{align}
\end{thm}
Identification of the model parameters comes from considering all
possible combinations of cutoffs. It is clear from Theorem \ref{thm:sufficiency}
that different choices for $\left(j,k,l,d_{0},d_{3}\right)$ reveal
information about distinct linear combinations of $\left(\rho,\gamma\right)$.
By considering multiple choices of $\left(j,k,l,d_{0},d_{3}\right)$,
and then aggregating the resulting information, we can identify all
the model parameters. We require an additional assumption before stating
our main identification result.
\begin{assumption}
\label{assu:XVariation}For all $\left(j,l\right)$ such that $2\leq j\leq k\leq l$,
and for all $d_{0},d_{3}\in\left\{ 0,1\right\} $
\[
Var\left(\left.\Delta X_{i}\right|X_{i,2}=X_{i,3},C_{j,l}\right)
\]
 is invertible.
\end{assumption}
This assumption guarantees that for each choice of $\left(j,l\right)$,
there is sufficient variation in $\Delta X_{i}$ in the subpopulation
of stayers $X_{i,2}=X_{i,3}$ to identify the regression coefficient.
This assumption can be weakened: we only need sufficient variation
for some $\left(j,l\right)$. However, if it fails for sufficiently
many $\left(j,l\right)$, identification of some of the threshold
parameters may fail.

Denote by $Y_{i}=\left(Y_{i,0},Y_{i,1},Y_{i,2},Y_{i,3}\right)$ the
time series of dependent variables for a given individual. 
\begin{thm}[Identification]
\label{thm:Identification}If Assumption 2 holds, then $\left(\beta,\rho,\gamma\right)$
can be identified from the joint distribution of the vector $\left(X_{i},Y_{i}\right)$
generated by the dynamic ordered logit model with fixed effects.
\end{thm}

\section{Estimation\label{sec:Estimation}}

Theorem \ref{thm:sufficiency} suggests that, for each choice of $2\leq j\leq k\leq l$,
we could use a conditional maximum likelihood estimator (CMLE) to
estimate a linear combination of the model parameters. Theorem \ref{thm:Identification}
suggests that a composite CMLE (CCMLE), based on the combination of
conditional likelihoods across all choices of $\left(j,k,l\right)$,
may be used to estimate the model parameters $\left(\beta,\rho,\gamma\right)$.
In this section, we define that CCMLE and establish conditions under
which it has desirable large sample properties. We focus on the discrete
regressor case. Results for continuous regressors can be obtained
by adapting Theorems 1 and 2 in HK to our case.

The binary random variable 
\[
C_{i,jl}=1\left\{ \left(D_{i,1}=0,D_{i,2}\left(l\right)=1\right)\text{ or }\left(D_{i,1}=1,D_{i,2}\left(j\right)=0\right)\right\} \times1\left\{ X_{i,2}=X_{i,3}\right\} .
\]
indicates whether $i$'s time series fits the description in $C_{j,l}=A_{j,l}\cup B_{j,l}$,
and that it is also a ``stayer'' in the sense that $X_{i2}=X_{i3}$.
Note that if $C_{i,jl}=1$, then $D_{i,1}=1$ implies that the individual
time series is of the type $B_{j,l}$. Similarly, if $C_{i,jl}=1$,
then $D_{i,1}=0$ implies that individual $i$ is of type $A_{j,l}$.

In the log-likelihood contribution below, (\ref{eq:CCML}), we have
substituted $D_{i,0}$ for $d_{0}$ in equation (\ref{eq:sufficiency_B}).
The value to substitute for $d_{3}$ depends on whether we are in
case $A$ or $B$. To that end, define
\begin{align*}
D_{i,3,jl} & =\begin{cases}
D_{i,3}\left(j\right) & \text{ if }D_{i,1}=0,\\
D_{i,3}\left(l\right) & \text{ if }D_{i,1}=1.
\end{cases}
\end{align*}
The conditional log likelihood contribution for individual $i$, for
cutoffs $\left(j,l\right),$ $2\leq j\leq k\leq l\leq J$, can then
be written

\begin{align}
l_{i,jl}\left(\beta,\rho,\gamma_{j},\gamma_{l}\right) & =C_{i,jl}\left[D_{i,1}\ln\left\{ \Lambda\left(\Delta X_{i}\beta+\rho\left(D_{i,0}-D_{i,3,jl}\right)+\gamma_{l}\left(1-D_{i,3,jl}\right)+\gamma_{j}D_{i,3,jl}\right)\right\} +\right.\nonumber \\
 & \phantom{}\left.\phantom{+}+\left(1-D_{i,1}\right)\ln\left\{ 1-\Lambda\left(\Delta X_{i}\beta+\rho\left(D_{i,0}-D_{i,3,jl}\right)+\gamma_{l}\left(1-D_{i,3,jl}\right)+\gamma_{j}D_{i,3,jl}\right)\right\} \right].\label{eq:CCML}
\end{align}
The CCMLE is 
\begin{equation}
\widehat{\theta}_{n}=\left(\widehat{\beta}_{n},\widehat{\rho}_{n},\widehat{\gamma}_{n}\right)=\arg\max\frac{1}{n}\sum_{2\leq j\leq k\leq l}\sum_{i=1}^{n}l_{i,jl}\left(\beta,\rho,\gamma_{j},\gamma_{l}\right),\label{eq:CCMLE}
\end{equation}
where we have implicitly imposed $\gamma_{k}=0$ in the definition
of $l_{i,jl}$.

We maintain the following assumption to establish the asymptotic properties
of the CCMLE.
\begin{assumption}[Stayers]
\label{assu:discrete_stayers} $P\left(X_{i,2}=X_{i,3}\right)>0.$
\end{assumption}
With additional technical work, this assumption can be relaxed to
the case where $X_{i,2}-X_{i,3}$ is continuously distributed with
positive density around zero, see HK's Theorem 1 and 2. 
\begin{thm}
\label{thm:asymptotics-CCMLE}Let \textup{$\left\{ \left(Y_{i},X_{i}\right),\,i=1,\cdots n\right\} $}
be a random sample of size $n$ from the dynamic ordered logit model
with fixed effects with true parameter values $\theta_{0}=\left(\beta_{0},\rho_{0},\gamma_{0}\right)$.
Under Assumptions \ref{assu:XVariation} and \ref{assu:discrete_stayers},
and for any value of $\theta_{0}$, 
\[
\widehat{\theta}_{n}\stackrel{p}{\to}\theta_{0}\text{ as }n\to\infty.
\]
Furthermore,
\[
\sqrt{n}\left(\widehat{\theta}_{n}-\theta_{0}\right)\stackrel{d}{\to}\mathcal{N}\left(0,H^{-1}\Sigma H^{-1}\right)\text{ as }n\to\infty,
\]
where $\Omega$ as the variance of the score of the composite likelihood,
defined in (\ref{eq:score_CCML}), and $H$ is the associated Hessian
defined in (\ref{eq:Hessian_CCML}).
\end{thm}
\begin{rem}
The convexity of the summands in (\ref{eq:CCMLE}) means that the
objective function is convex. We compute the CCMLE using the Newton-Raphson
algorithm in R's nlm function (R Core Team, 2020). Supplying analytical
gradients and Hessians speeds up the estimation. 
\end{rem}

\section{Persistence in self-reported health status\label{sec:health}}

Our analysis uses panel data for the period 2003-2016 from the European
Union Statistics on Income and Living Conditions (EU-SILC), see Eurostat
(2017) for detailed documentation. The microdata is publicly available
upon request.\footnote{The data were made available to us by Eurostat under Contract RPP
132-2018-EU-SILC.} EU-SILC provides a set of indicators on income and poverty, social
inclusion, living conditions and, importantly, health status. For
each country in the European Union, plus Iceland, Norway, and Switzerland,
EU-SILC contains data on a representative sample of the population
of those 18 years and older. 

EU-SILC is a rotating panel. Every individual is followed over a period
of two to four years. The total number of individual-years for the
period 2003-2016 is 1273877. Our identification result demands four
observations per individual, so we restrict attention to individuals
that report valid information on their health status for 4 consecutive
years. This restriction, and the restriction that the explanatory
variables that we use in the analysis below have non-missing information,
leaves us with a sample of 260601 individuals, for 1042404 individual-years.
The proportion of incomplete samples differs across countries. As
a result, the sample we work with may not be representative of EU-SILC's
population. For example, out of the 27 countries that contribute to
our sample, the largest contributors are Italy (with 43385 individuals),
Spain (25634), and Poland (22628); the smallest are Portugal (12),
Iceland (1496), and Slovakia (1982).

The outcome variable in our analysis is self-reported health status:
self-perceived physical health, elicited during EU-SILC interviews.
The person answers the question on how she perceives her physical
health to be in general, at the date of the survey, by classifying
it as one of: (1) bad and very bad (12\% in our sample); (2) fair
(26\%); (3) good (44\%); (4) very good (19\%).\footnote{We have merged the separate categories ``bad'' and ``very bad''
in the original reported variable, because there is only a small fraction
of observations with ``very bad'' health status.} Out of $260601\times3=781803$ health transitions that we observe,
most often there is no change in health status (65.6\%). Decreases
by one unit (16\%) are slightly more frequent than increases by one
unit (15\%). Two-unit increases (1.2\%) and decreases (1.5\%) are
infrequent, and three-unit increases (0.08\%) and decreases (0.12\%)
are rare. 

Table \ref{tab:health-income} relates the outcome variable, and changes
to the outcome variable, to our main explanatory variable of interest,
log income (total disposable household equivalised income). Household
income was scaled using the composition and size of each household.
This scale is based on the OECD modified equivalence scale, which
gives a weight of 1.0 to the first adult in the household, 0.5 to
other adults and 0.3 to each child (under 14 years old). 

The table provides descriptive statistics for log income in our sample,
grouped by health status. The top panel is in levels. Average income
is increasing in health status. The bottom panel is in changes, which
represents one way to control for unobserved heterogeneity. The implied
increases for changes are close to zero, hinting at the imported role
of unobserved heterogeneity.
\begin{table}
\centering{}%
\begin{tabular}{c>{\centering}p{2cm}rr}
\hline 
 &  & \multicolumn{2}{c}{Log income}\tabularnewline
\hline 
 &  & mean & sd\tabularnewline
\hline 
health status & 1 & 8.68 & 0.92\tabularnewline
 & 2 & 8.95 & 0.94\tabularnewline
 & 3 & 9.29 & 0.94\tabularnewline
 & 4 & 9.51 & 0.90\tabularnewline
\hline 
 &  & \multicolumn{2}{c}{$\Delta$Log income}\tabularnewline
 &  & mean & sd\tabularnewline
 & -3 & 0.05 & 0.47\tabularnewline
$\Delta$health status & -2 & 0.06 & 0.43\tabularnewline
 & -1 & 0.07 & 0.40\tabularnewline
 & 0 & 0.08 & 0.38\tabularnewline
 & 1 & 0.08 & 0.40\tabularnewline
 & 2 & 0.08 & 0.43\tabularnewline
 & 3 & 0.07 & 0.49\tabularnewline
\hline 
\end{tabular}\caption{Health and income}
\label{tab:health-income}
\end{table}

Table \ref{tab:summary_stats} reports a set of descriptive statistics
on income and other explanatory variables, described in the next few
paragraphs. In our analysis below, we control for some time-varying
variables that are standard in the literature. First, the number of
children is measured as the number of persons living in the private
household that are age $14$ or less, top-coded at 3 children. In
our sample, 74\% of respondents have no children, and the average
number of children is 0.40. Second, marriage status is a dummy variable
that indicates being married or living together. The majority of the
individuals are married (61\%). Third, we use a self-reported indicator
for labor market status variable that we map onto 4 values: (1) employed,
51\%; (2) unemployed, 5.2\%; (3) retired, 12\% and (4) other, 32\%.
The value ``other'' includes students, permanently disabled or unfit
to work, and fulfilling domestic tasks and care responsibilities.
\begin{table}
\centering{}%
\begin{tabular}{llrr}
\hline 
 &  & mean & sd\tabularnewline
\hline 
health status & (1) bad and very bad & 0.115 & \tabularnewline
 & (2) fair & 0.260 & \tabularnewline
 & (3) good & 0.437 & \tabularnewline
 & (4) very good & 0.187 & \tabularnewline
\hline 
\multicolumn{4}{l}{\emph{Time-varying explanatory variables}}\tabularnewline
log income &  & 9.172 & 0.965\tabularnewline
child &  & 0.401 & 0.755\tabularnewline
married &  & 0.614 & \tabularnewline
employment status & employed & 0.512 & \tabularnewline
 & unemployed & 0.052 & \tabularnewline
 & retired & 0.117 & \tabularnewline
 & other & 0.319 & \tabularnewline
\hline 
\multicolumn{4}{l}{\emph{Time-invariant explanatory variables}}\tabularnewline
age group & $[18;25]$ & 0.082 & \tabularnewline
 & $]25;35]$ & 0.145 & \tabularnewline
 & $]35;45]$ & 0.188 & \tabularnewline
 & $]45;55]$ & 0.194 & \tabularnewline
 & $]55;65]$ & 0.182 & \tabularnewline
 & $]65;\infty]$ & 0.208 & \tabularnewline
urbanisation & high & 0.388 & \tabularnewline
 & middle & 0.224 & \tabularnewline
 & low & 0.388 & \tabularnewline
male &  & 0.460 & \tabularnewline
educ & no schooling & 0.013 & \tabularnewline
 & primary & 0.143 & \tabularnewline
 & lower secondary & 0.188 & \tabularnewline
 & upper secondary & 0.420 & \tabularnewline
 & post-secondary & 0.038 & \tabularnewline
 & tertiary & 0.199 & \tabularnewline
\hline 
$n$ &  & 260601 & \tabularnewline
$T$ &  & 4 & \tabularnewline
$nT$ &  & 1042404 & \tabularnewline
\hline 
\end{tabular}\caption{Descriptive statistics}
\label{tab:summary_stats}
\end{table}

In our fixed effects results below, we do not further control for
variables that do not change over the sample period for a given individual.
However, we include a set of time-invariant explanatory variables
when we obtain results for non-fixed effects estimators.\footnote{Coefficient estimates for these variables are omitted from the main
text, and reported in Appendix \ref{sec:Additional-empirical-results}.} Table \ref{tab:summary_stats} provides descriptive statistics for
such variables. The total sample contains slightly more females (54\%)
than males (46\%). The proportion of individuals aged between 18 and
25 is 8.2\%; 21\% of individuals are aged 65 or more. With regards
to education, 1.3\% of the sample have no schooling (0); 14\% have
attended primary school (1); 19\% have lower secondary education (3);
42\% have upper secondary education (4); 3.8\% have post-secondary
education (5) and 20\% have tertiary education (6). Geographically,
39\% of individuals live in areas with a high degree of urbanisation
and 39\% live in areas with low levels of urbanisation. 

We estimate the parameters in the dynamic ordered choice model with
fixed effects, with latent variable outcome equation
\begin{align}
SRH_{i,t}^{*} & =\alpha_{i}+\rho1\left\{ SRH_{i,t-1}\geq3\right\} +\beta_{1}\log income_{it}+\beta_{2}child_{it}+\beta_{3}married_{it}+\label{eq:dofe_health}\\
 & \phantom{=}+\beta_{4}unemp_{it}+\beta_{5}retired_{it}+\beta_{6}other_{it}-U_{it}.\nonumber 
\end{align}
Regression results are presented in Table \ref{tab:DOFE_results}.
The first four columns (a-d, ``DOLFE'', for \emph{d}ynamic \emph{o}rdered
\emph{l}ogit with \emph{f}ixed \emph{e}ffects) presents the results
for (\ref{eq:dofe_health}) using the estimator described in Section
\ref{sec:Estimation}. Different values of $h$ refer to a bandwidth
parameter that we introduce because one of the explanatory variables
is continuous, as in HK. Column (d) omits the employment variables,
to check whether relationship between employment status and income
matters for estimation of the effect of income on health. 

We also present estimation results for different estimators. Results
for the static ordered logit model with fixed effects, i.e. setting
$\rho=0$ in (\ref{eq:dofe_health}), are obtained using the estimator
in Muris (2017), and presented in column (e) (``FEOL''). Column
(f) (``DOL'') estimates a dynamic ordered logit model without fixed
effects, i.e. (\ref{eq:dofe_health}) with $\alpha_{i}=0$. Column
(g) (``OL'') presents results for cross-sectional ordered logit
estimator that does not take into account fixed effects or dynamics
(i.e. $\alpha_{i}=\rho=0$ in (\ref{eq:dofe_health})). We also present
results for a static linear model with (h, ``FELM'') and without
(i, ``LM'') fixed effects. The standard errors for all estimators
are obtained using the bootstrap (500 replications). For the estimators
that are not of the fixed effects type, we additionally control for
education, gender, education level and the level of urbanisation.
DOLFE uses four periods of data, corresponding to $t=0,1,2$. For
comparability, the other dynamic estimator also uses periods 0,1,2;
static estimators use periods 1,2.
\begin{table}
\begin{centering}
{\footnotesize{}\hspace*{-1cm}}%
\begin{tabular}[b]{lccccccccc}
\hline 
 & {\footnotesize{}(a)} & {\footnotesize{}(b)} & {\footnotesize{}(c)} & {\footnotesize{}(d)} & {\footnotesize{}(e)} & {\footnotesize{}(f)} & {\footnotesize{}(g)} & {\footnotesize{}(h)} & {\footnotesize{}(i)}\tabularnewline
 & {\footnotesize{}DOLFE} & {\footnotesize{}DOLFE} & {\footnotesize{}DOLFE} & {\footnotesize{}DOLFE} & {\footnotesize{}FEOL} & {\footnotesize{}DOL} & {\footnotesize{}OL} & {\footnotesize{}FELM} & {\footnotesize{}LM}\tabularnewline
 & {\footnotesize{}$h=1$} & {\footnotesize{}$h=0.1$} & {\footnotesize{}$h=10$} & {\footnotesize{}$h=1$} &  &  &  &  & \tabularnewline
\hline 
{\footnotesize{}log(income)} & {\footnotesize{}0.049} & {\footnotesize{}-0.047} & {\footnotesize{}0.059} & {\footnotesize{}0.061} & {\footnotesize{}0.020} & {\footnotesize{}0.340} & {\footnotesize{}0.492} & {\footnotesize{}0.003} & {\footnotesize{}0.194}\tabularnewline
\multirow{1}{*}[103cm]{} & \multirow{1}{*}{{\tiny{}(0.033)}} & {\tiny{}(0.056)} & {\tiny{}(0.029)} & {\tiny{}(0.029)} & {\tiny{}(0.019)} & {\tiny{}(0.004)} & {\tiny{}(0.004)} & {\tiny{}(0.003)} & {\tiny{}(0.002)}\tabularnewline
{\footnotesize{}child} & {\footnotesize{}-0.030} & {\footnotesize{}0.006} & {\footnotesize{}-0.031} & {\footnotesize{}-0.026} & {\footnotesize{}0.021} & {\footnotesize{}0.060} & {\footnotesize{}0.089} & {\footnotesize{}0.002} & {\footnotesize{}0.033}\tabularnewline
 & {\tiny{}(0.051)} & {\tiny{}(0.069)} & {\tiny{}(0.050)} & {\tiny{}(0.049)} & {\tiny{}(0.032)} & {\tiny{}(0.005)} & {\tiny{}(0.005)} & {\tiny{}(0.005)} & {\tiny{}(0.002)}\tabularnewline
{\footnotesize{}married} & {\footnotesize{}0.139} & {\footnotesize{}-0.041} & {\footnotesize{}0.157} & {\footnotesize{}0.130} & {\footnotesize{}0.164} & {\footnotesize{}0.073} & {\footnotesize{}0.141} & {\footnotesize{}0.029} & {\footnotesize{}0.062}\tabularnewline
 & {\tiny{}(0.087)} & {\tiny{}(0.119)} & {\tiny{}(0.086)} & {\tiny{}(0.088)} & {\tiny{}(0.053)} & {\tiny{}(0.007)} & {\tiny{}(0.008)} & {\tiny{}(0.009)} & {\tiny{}(0.003)}\tabularnewline
{\footnotesize{}unemp} & {\footnotesize{}-0.188} & {\footnotesize{}-0.230} & {\footnotesize{}-0.178} &  & {\footnotesize{}-0.196} & {\footnotesize{}-0.242} & {\footnotesize{}-0.308} & {\footnotesize{}-0.033} & {\footnotesize{}-0.127}\tabularnewline
 & {\tiny{}(0.070)} & {\tiny{}(0.110)} & {\tiny{}(0.068)} &  & {\tiny{}(0.038)} & {\tiny{}(0.014)} & {\tiny{}(0.015)} & {\tiny{}(0.007)} & {\tiny{}(0.006)}\tabularnewline
{\footnotesize{}retired} & {\footnotesize{}-0.132} & {\footnotesize{}-0.043} & {\footnotesize{}-0.139} &  & {\footnotesize{}-0.154} & {\footnotesize{}-0.050} & {\footnotesize{}-0.097} & {\footnotesize{}-0.027} & {\footnotesize{}-0.047}\tabularnewline
 & {\tiny{}(0.082)} & {\tiny{}(0.119)} & {\tiny{}(0.080)} &  & {\tiny{}(0.041)} & {\tiny{}(0.010)} & {\tiny{}(0.011)} & {\tiny{}(0.007)} & {\tiny{}(0.004)}\tabularnewline
{\footnotesize{}other} & {\footnotesize{}-0.370} & {\footnotesize{}-0.207} & {\footnotesize{}-0.369} &  & {\footnotesize{}-0.473} & {\footnotesize{}-0.771} & {\footnotesize{}-1.087} & {\footnotesize{}-0.082} & {\footnotesize{}-0.460}\tabularnewline
 & {\tiny{}(0.061)} & {\tiny{}(0.087)} & {\tiny{}(0.061)} &  & {\tiny{}(0.040)} & {\tiny{}(0.010)} & {\tiny{}(0.012)} & {\tiny{}(0.007)} & {\tiny{}(0.005)}\tabularnewline
\hline 
{\footnotesize{}$\rho$} & {\footnotesize{}0.733} & {\footnotesize{}0.723} & {\footnotesize{}0.733} & {\footnotesize{}0.734} &  & {\footnotesize{}1.987} &  &  & \tabularnewline
 & {\tiny{}(0.020)} & {\tiny{}(0.025)} & {\tiny{}(0.020)} & {\tiny{}(0.017)} &  & {\tiny{}(0.023)} &  &  & \tabularnewline
{\footnotesize{}$\gamma_{2}$} & {\footnotesize{}-3.275} & {\footnotesize{}-3.260} & {\footnotesize{}-3.272} & {\footnotesize{}-3.211} & {\footnotesize{}-3.487} & {\footnotesize{}-2.506} & {\footnotesize{}-1.992} &  & \tabularnewline
 & {\tiny{}(0.054)} & {\tiny{}(0.068)} & {\tiny{}(0.053)} & {\tiny{}(0.048)} & {\tiny{}(0.015)} & {\tiny{}(0.007)} & {\tiny{}(0.006)} &  & \tabularnewline
{\footnotesize{}$\gamma_{4}$} & {\footnotesize{}3.326} & {\footnotesize{}3.356} & {\footnotesize{}3.329} & {\footnotesize{}3.321} & {\footnotesize{}3.997} & {\footnotesize{}3.089} & {\footnotesize{}2.603} &  & \tabularnewline
 & {\tiny{}(0.055)} & {\tiny{}(0.076)} & {\tiny{}(0.055)} & {\tiny{}(0.054)} & {\tiny{}(0.024)} & {\tiny{}(0.006)} & {\tiny{}(0.006)} &  & \tabularnewline
\hline 
\end{tabular}{\footnotesize\par}
\par\end{centering}
\caption{Main results.}

\centering{}\label{tab:DOFE_results}
\end{table}

\textbf{Income. }Our main explanatory variable of interest is income
(log income, coefficient $\beta_{1}$). Across almost all specifications,
we find a positive association between income and self-reported health.
The only exception is column (b), where the point estimate is negative,
and about the same magnitude as the standard error. 

Controlling for unobserved heterogeneity leads to a very strong reduction
in the magnitude of the association. For example, for the static case,
a comparison of columns (e) and (g) says that, for the static case,
controlling for unobserved heterogeneity reduces the coefficient on
income by more than a factor 20. For this comparison, note that the
threshold differences increase, suggesting that the scale increases;
compare also the coefficients on the other variables, with an unchanged
order of magnitude. We are not the first to observe a limited association
between income and self-reported health. In a review of the literature,
Gunasekara et al. (2011) found a small positive link between income
and self-reported health, which is reduced when controlling for unmeasured
confounders. Interestingly, Johnston et al. (2009) found no link between
self-reported hypertension and income; an association that, however,
became positive when using objective measures of hypertension. 

The estimated effect of income also changes when we control for state
dependence. Comparing columns (f) and (g), we see that controlling
for state dependence in a model without unobserved heterogeneity reduces
the association between income and self-reported health. So, individually
controlling for unobserved heterogeneity or for dynamics reduces the
magnitude of the association between health and income. 

Finally, a comparison between columns (a) and (d) shows that the estimate
for income association is robust to controlling for employment status,

\textbf{State dependence. }We estimate an autoregressive parameter
of around 0.75, with threshold differences of about $3$. The estimated
ratio of $\rho$ to the thresholds (which measure the distance from
category 3) are much lower than for column (f). This confirms the
importance of controlling for unobserved heterogeneity, which reduces
the estimated magnitude of persistence by a factor 3. Nevertheless,
even when controlling for unobserved (and observed) heterogeneity,
we find strong evidence for large, positive persistence in self-reported
health.

There are at least two ways to get a sense of the magnitude of persistence.
The first approach, also available for binary choice methods, is to
compare estimates of $\rho$ to estimates of regression coefficients.
For example, in our preferred specification in column (a), a health
shock that lifts you from any category below 3, to category 3 or 4,
has an impact on future health that is almost 4 times that of becoming
unemployed. The impact is more than 5 times that of marrying.

The second approach to interpreting estimates of $\rho$ uses the
estimated thresholds to obtain an estimate similar to a linear autoregressive
model.\footnote{This approach is not available for binary choice models because threshold
parameters are not available. } Differences between the thresholds are a measure of the distances
between two categories. If $\gamma_{2}=-\gamma_{4}$, then categories
2 and 3 are as far apart as categories 3 and 4. In such a case, a
linear model may yield similar results in terms of partial effects.
In this case, $-\rho/\gamma_{2}$ and $\rho/\gamma_{3}$ can be interpreted
as linear regression coefficients for that category; we find that
they are about 0.25. Said differently, we find the analog of an AR(1)
coefficient of 0.25 in a linear model.

\textbf{Other time-varying covariates.} The literature so far has
been inconclusive on how retirement is associated with health. On
one hand, retiring allows more time for health-promoting activities,
and reduces work-related stress. On the other hand, people may lose
traction and motivation and may become less active. Therefore, while
Coe and Zamarro (2011) find that retirement improves health, Behncke
(2012) finds an increase in the likelihood of disease following retirement.
In our DOLFE model, the coefficient is statistically insignificant.
This suggests that the association between retirement and health may
not be as strong as previously thought. Compared to the FEOL model,
the effect of retirement on health disappears when controlling for
state dependence. 

The extensive literature on the link between unemployment and health
in particular, and economic conditions and health more generally,
is broadly inconclusive. Some studies have suggested a protective
role of unemployment on health (Ruhm, 2000), while others suggest
that unemployment is detrimental for health (McInerney and Mellor,
2012). Our results appear to be more in line with the findings of
Ruhm (2015) and Böckerman and Ilmakunnas (2009). In DOLFE, the coefficient
of being unemployed is negative and statistically significant. Controlling
for state dependence does not change things compared to the FEOL model. 

Having children is insignificant in our DOLFE model, while previous
studies have provided mixed findings on this question (Mckenzie and
Carter, 2013; Evenson and Simon, 2005). This is also insignificant
in the FEOL model, suggesting that previous findings on having children
might have been driven by unobserved heterogeneity. 

Being married is generally considered a protective factor for health
(Kaplan and Kronick, 2006; Molloy et al., 2009). In our model, however,
it is statistically insignificant -- as opposed to the FEOL model
where it was positive and significant. Thus, controlling for state
dependence appears to be important for this variable. 

\section{Conclusion}

This paper studies a fixed$-T$ dynamic ordered logit model with fixed
effects (DOLFE) and is the first to provide identification and estimation
results for all common parameters in a dynamic ordered logit model
with fixed effects and a fixed number of time periods. The results
require only four time periods of data on the ordinal outcome variable.
We demonstrate identification of the autoregressive coefficients on
the lagged dependent variable, the regression coefficients on the
exogenous regressors, and differences of the threshold parameters.
The latter makes it possible to interpret the magnitude of the coefficients.

Including fixed effects and state dependence in the model is particularly
relevant for self-reported health, a measure that is widely used in
the literature. Future research using self-reporting health can benefit
from our model for two main reasons. First, controlling for fixed
effects, one can take into account unobserved heterogeneity (Carro
and Traferri, 2014; Halliday, 2008; Fernández-Val et al., 2017), which
is especially important due to differences in understanding and reporting
health status (Groot, 2000; Sen, 2002; Jylhä et al., 1998; Baron-Epel
et al., 2005; Jürges, 2007). Second, it incorporates elements of persistence
(Contoyannis et al., 2004; Ohrnberger et al., 2017; Hernández-Quevedo
et al., 2008; Roy and Schurer, 2013) or adaptation (Cubí-Mollá et
al., 2017; Daltroy et al., 1999; Damschroder et al., 2005; Heiss et
al., 2014) by controlling for state dependence (Carro and Traferri,
2014; Fernández-Val et al., 2017; Halliday, 2008). Thus, using our
estimator addresses such biases often present in studies using self-rated
health (Davillas et al., 2017). 

We thus applied the new dynamic ordered logit model with fixed effects
to investigate the determinants of self-reported health, focusing
on the link between income and health in a panel of European countries.
We found that when controlling for unobserved heterogeneity, the association
between income and health becomes statistically insignificant. This
is in line with studies that have found a smaller or insignificant
association when using fixed effects (Gunasekara, 2011; Larrimore,
2011) -- while other studies have suggested a positive association
between the two (Carrieri and Jones, 2017; Ettner, 1996; Frijters
et al., 2005; Mackenbach et al., 2005). Being retired or married also
becomes statistically insignificant in our model when controlling
for state dependence. Being unemployed or having children does not
appear to be associated with self-reported health in our model. 

Our empirical results suggest that persistence plays a positive and
significant role in one's self-reported health. In other words, one's
health is dependent on the health in the previous period, which is
a reasonable thing to expect, as health problems may expand over a
number of periods, or become permanent. This element reflects persistence
of health status over time (Contoyannis et al., 2004). Furthermore,
what is particularly interesting is that in our data, self-reported
health tends to improve, on average, over time - even as people become
four years older during the study period - and being older is typically
associated with worse health outcomes. Therefore, it is reasonable
to believe that this improvement in self-reported health is often
subjective, and does not necessarily reflect one's objective health
level. This element might reflect adaptation to health problems: Even
though one's health does not improve, they adapt to their situation
and therefore report better health (Cubí-Mollá et al 2017; Daltroy
et al., 1999; Damschroder et al., 2005). This second element is a
typical bias when using self-reported health outcomes, and our model
helps correct such biases by introducing the dynamic element to a
fixed effects ordered model. Another interesting finding is that,
when controlling for unobserved heterogeneity, the link between income
and health becomes statistically insignificant, suggesting that other
factors might explain the association between the two. 

Overall, measurement bias in studies using self-reported outcomes
often poses challenges to research, that may discourage the use of
such variables. Our model addresses these biases, and thus provides
a basis for more choices in conducting research with databases that
provide such variables.

\appendix

\section{Proofs\label{sec:Proofs}}
\begin{proof}[Proof of Theorem \ref{thm:sufficiency}.]
 Recall the definition of the events in the main text, for $2\leq j\leq k\leq l\leq J,$
\begin{align*}
A_{j,l} & =\left\{ D_{i,0}\left(k\right)=d_{0},D_{i,1}\left(k\right)=0,D_{i,2}\left(l\right)=1,D_{i,3}\left(j\right)=d_{3}\right\} \\
B_{j,l} & =\left\{ D_{i,0}\left(k\right)=d_{0},D_{i,1}\left(k\right)=1,D_{i,2}\left(j\right)=0,D_{i,3}\left(l\right)=d_{3}\right\} ,\\
C_{j,l} & =A_{j,l}\cup B_{j,l}.
\end{align*}
The following derivations modifies the development in HK (p. 843-844).
The probability of the event $B_{j,l}$, conditional on the covariates
$X_{i}$ and the unobserved heterogeneity $\alpha_{i}$ is given by
\begin{align}
P\left(\left.B_{j,l}\right|X_{i},\alpha_{i}\right) & =p_{0}\left(X_{i},\alpha_{i}\right)^{d_{0}}\left[1-p_{0}\left(X_{i},\alpha_{i}\right)\right]^{1-d_{0}}\label{eq:pAjl}\\
 & \phantom{=}\times\frac{\exp\left(\alpha_{i}+X_{i,1}\beta+\rho d_{0}\right)}{1+\exp\left(\alpha_{i}+X_{i,1}\beta+\rho d_{0}\right)}\nonumber \\
 & \phantom{=}\times\frac{1}{1+\exp\left(\alpha_{i}+X_{i,2}\beta+\rho-\gamma_{j}\right)}\nonumber \\
 & \phantom{=}\times\frac{\left[\exp\left(\alpha_{i}+X_{i,3}\beta-\gamma_{l}\right)\right]^{d_{3}}}{1+\exp\left(\alpha_{i}+X_{i,3}\beta-\gamma_{l}\right)}.\nonumber 
\end{align}
Similarly, for $A_{j,l}$, 
\begin{align}
P\left(\left.A_{j,l}\right|X_{i},\alpha_{i}\right) & =p_{0}\left(X_{i},\alpha_{i}\right)^{d_{0}}\left[1-p_{0}\left(X_{i},\alpha_{i}\right)\right]^{1-d_{0}}\label{eq:pBjl}\\
 & \phantom{=}\times\frac{1}{1+\exp\left(\alpha_{i}+X_{i,1}\beta+\rho d_{0}\right)}\nonumber \\
 & \phantom{=}\times\frac{\exp\left(\alpha_{i}+X_{i,2}\beta-\gamma_{l}\right)}{1+\exp\left(\alpha_{i}+X_{i,2}\beta-\gamma_{l}\right)}\nonumber \\
 & \phantom{=}\times\frac{\left[\exp\left(\alpha_{i}+X_{i,3}\beta+\rho-\gamma_{j}\right)\right]^{d_{3}}}{1+\exp\left(\alpha_{i}+X_{i,3}\beta+\rho-\gamma_{j}\right)}.\nonumber 
\end{align}
The probability of event $A_{j,l}$, conditional on the event $C_{j,l}=A_{j,l}\cup B_{j,l}$
and on $X_{i,2}=X_{i,3}$ is given by
\begin{align}
P\left(\left.A_{j,l}\right|X_{i},\alpha_{i},C_{j,l},X_{i,2}=X_{i,3}\right) & =\frac{P\left(\left.A_{j,l},A_{j,l}\cup B_{j,l}\right|X_{i},\alpha_{i},X_{i,2}=X_{i,3}\right)}{P\left(\left.A_{j,l}\cup B_{j,l}\right|X_{i},\alpha_{i},X_{i,2}=X_{i,3}\right)}\nonumber \\
 & =\frac{P\left(\left.A_{j,l}\right|X_{i},\alpha_{i},X_{i,2}=X_{i,3}\right)}{P\left(\left.A_{j,l}\right|X_{i},\alpha_{i},X_{i,2}=X_{i,3}\right)+P\left(\left.B_{j,l}\right|X_{i},\alpha_{i},X_{i,2}=X_{i,3}\right)}\nonumber \\
 & =\frac{1}{1+P\left(\left.B_{j,l}\right|X_{i},\alpha_{i},X_{i,2}=X_{i,3}\right)/P\left(\left.A_{j,l}\right|X_{i},\alpha_{i},X_{i,2}=X_{i,3}\right)}.\label{eq:final_ratio}
\end{align}
where the first step follows from the definition of conditional probability;
the second follows from the fact that $A_{j,l}$ and $B_{j,l}$ are
disjoint; the third from division by the probability of $A_{j,l}$.
Plugging the conditional probabilities (\ref{eq:pAjl}) and (\ref{eq:pBjl})
into the final expression (\ref{eq:final_ratio}) in the display below
obtains our final sufficiency result:
\begin{align}
P\left(\left.A_{j,l}\right|X_{i},C_{j,l},X_{i,2}=X_{i,3}\right) & =\frac{1}{1+\exp\left(\Delta X_{i}\beta+\rho\left(d_{0}-d_{3}\right)+\left(1-d_{3}\right)\gamma_{l}+d_{3}\gamma_{j}\right)},\label{eq:sufficiency_A-1}\\
P\left(\left.B_{j,l}\right|X_{i},C_{j,l},X_{i,2}=X_{i,3}\right) & =\frac{\exp\left(\Delta X_{i}\beta+\rho\left(d_{0}-d_{3}\right)+\left(1-d_{3}\right)\gamma_{l}+d_{3}\gamma_{j}\right)}{1+\exp\left(\Delta X_{i}\beta+\rho\left(d_{0}-d_{3}\right)+\left(1-d_{3}\right)\gamma_{l}+d_{3}\gamma_{j}\right)}.\label{eq:sufficiency_B-1}
\end{align}
\end{proof}
\begin{proof}[Proof of Theorem \ref{thm:Identification}]
For notational convenience, we refer to the conditional probability
obtained in our sufficiency result, Theorem \ref{thm:sufficiency},
equation~(\ref{eq:sufficiency_B-1}), as 
\begin{align}
p_{jl}\left(X_{i},d_{0},d_{3}\right) & =P\left(\left.B_{j,l}\right|X_{i},C_{j,l},X_{i,2}=X_{i,3}\right)\nonumber \\
 & =\Lambda\left(\Delta X_{i}\beta+\rho\left(d_{0}-d_{3}\right)+\left(1-d_{3}\right)\gamma_{l}+d_{3}\gamma_{j}\right).\label{eq:sufficiency_ID}
\end{align}
Evaluate this for the case of $j=k=l$ and for $d_{0}=d_{3}=0$,
\[
p_{kk}\left(X_{i},0,0\right)=\Lambda\left(X_{i}\beta\right),
\]
which is a simplification of (\ref{eq:sufficiency_ID}) because the
second and fourth term are zero due to $d_{0}=d_{3}=0$ and the third
term is zero because of the choice of $l$ and the scale normalization,
$\gamma_{l}=\gamma_{k}=0$. Then
\[
\beta=E_{kk,00}\left[\Delta X_{i}^{'}\Delta X_{i}\right]^{-1}E_{kk,00}\left[\Delta X_{i}^{'}\Lambda^{-1}\left(p_{kk}\left(X_{i},0,0\right)\right)\right],
\]
where $E_{jl,00}$ is the expectation conditional on $\left(X_{i},C_{j,l},X_{i,2}=X_{i,3}\right)$
and using the cutoffs $j=l=k$ and starting and ending values $d_{0}=d_{3}=0$.
The invertibility of the first term is due to Assumption \ref{assu:XVariation},
and the second term is well-defined because $p_{kk}$ is bounded away
from 0 and 1 because of the logistic errors. This obtains identification
of $\beta$.

Next, note that
\[
p_{kk}\left(X_{i},1,0\right)=\Lambda\left(\Delta X_{i}\beta+\rho\right).
\]
From this we obtain
\[
\rho=E_{kk,10}\left[\Lambda^{-1}\left(p_{kk}\left(X_{i},1,0\right)\right)-\Delta X_{i}\beta\right],
\]
where $E_{kk,10}$ now uses starting and ending values $d_{0}=1$
and $d_{3}=0$. This identifies the persistence parameter, since $\beta$
was identified previously.

To identify the thresholds $\gamma_{l},\,l>k$, consider that for
all $l>k$: 
\begin{align*}
p_{kl}\left(X_{i},0,0\right) & =\Lambda\left(\Delta X_{i}\beta+\gamma_{l}\right),\\
\gamma_{l} & =E_{kl,00}\left[\Lambda^{-1}\left(p_{kl}\left(X_{i},0,0\right)\right)-\Delta X_{i}\beta\right].
\end{align*}
Finally, to identify the thresholds $\gamma_{j},\,j<k$, consider
that for all $j<k$,
\begin{align*}
p_{jk}\left(X_{i},1,1\right) & =\Lambda\left(\Delta X_{i}\beta+\gamma_{j}\right),\\
\gamma_{j} & =E_{jk,11}\left[\Lambda^{-1}\left(p_{jk}\left(X_{i},1,1\right)\right)-\Delta X_{i}\beta\right].
\end{align*}
\end{proof}
\begin{proof}[Proof of Theorem \ref{thm:asymptotics-CCMLE}.]
\textbf{Consistency. }We will use the fact that the objective function
is concave (demonstrated in the next paragraph). We can therefore
use Theorem 2.7 in Newey and McFadden. That condition (i, identification)
holds is suggested by our identification result in Theorem \ref{thm:Identification}.
The information inequality and Assumption \ref{assu:XVariation} ensure
that identification is not lost when moving to the composite conditional
likelihood function, see also the Hessian below. Condition (iii, pointwise
convergence) follows from a law of large numbers for i.i.d. data. 

To see that the objective function is concave, denote $Z_{ijl}=\left(\Delta X_{i},D_{i,0}-D_{i,3,jl},\left(1-D_{i,3,jl}\right),D_{i,3,jl}\right)$
and $\theta_{jl}=\left(\beta,\rho,\gamma_{l},\gamma_{j}\right)$,
so that 
\begin{equation}
l_{i,jl}\left(\theta_{jl}\right)=C_{i,jl}\left[D_{i,1}\ln\Lambda\left(Z_{i,jl}\theta_{jl}\right)+\left(1-D_{i,1}\right)\ln\left[1-\Lambda\left(Z_{i,jl}\theta_{jl}\right)\right]\right],\label{eq:likelihood_theta}
\end{equation}
so that the score contribution is 
\begin{equation}
s_{i,jl}\left(\theta_{jl}\right)=C_{i,jl}\left[D_{i,1}-\Lambda\left(Z_{i,jl}\theta_{jl}\right)\right]Z_{i,jl}^{'}\label{eq:score_ijl}
\end{equation}
and the contribution to the Hessian is 
\begin{equation}
H_{i,jl}\left(\theta_{jl}\right)=-C_{i,jl}\Lambda\left(Z_{i,jl}\theta_{jl}\right)\left(1-\Lambda\left(Z_{i,jl}\theta_{jl}\right)\right)Z_{i,jl}^{'}Z_{i,jl}.\label{eq:Hessian_ijl}
\end{equation}
It can be seen immediately from (\ref{eq:Hessian_ijl}) that $l_{i,jl}\left(\theta_{jl}\right)$
is concave: $C_{i,jl}\in\left\{ 0,1\right\} $, $\Lambda\left(Z_{i,jl}\theta_{jl}\right)\left(1-\Lambda\left(Z_{i,jl}\theta_{jl}\right)\right)\in\left(0,1\right)$,
and $Z_{i,jl}^{'}Z_{i,jl}$ is positive semi-definite. Because sums
of concave functions are concave, the objective function 
\[
l_{n}\left(\beta,\rho,\gamma\right)=\sum_{2\leq j\leq k\leq l}l_{n,jl}\left(\beta,\rho,\gamma_{j},\gamma_{l}\right)
\]
is concave. This completes the proof of concavity.

\textbf{Asymptotic normality.} To demonstrate asymptotic normality
of the estimator, we will verify the conditions in Theorem 3.1 of
NM94. Condition (i, interior) holds by construction. The fact that
condition (ii, twice CD) holds can be seen by inspecting the expression
of the second derivative in (\ref{eq:Hessian_ijl}). Since the composite
conditional likelihood function is a sum of functions of that form,
it is also twice continuously differentiable. Condition (iii, CLT
for score) holds because standard central limit theorems for i.i.d.
data apply to (\ref{eq:score_ijl}). To see this, note that
\begin{align}
Var\left[s_{i,jl}\left(\theta_{jl,0}\right)\right] & =E\left[s_{i,jl}\left(\theta_{jl,0}\right)s_{i,jl}\left(\theta_{jl,0}\right)^{'}\right]\nonumber \\
 & =E\left[C_{i,jl}\left[D_{i,1}-\Lambda\left(Z_{i,jl}\theta_{jl}\right)\right]^{2}Z_{i,jl}^{'}Z_{i,jl}\right]\nonumber \\
 & =E\left[\left.E\left[C_{i,jl}\left[D_{i,1}-\Lambda\left(Z_{i,jl}\theta_{jl}\right)\right]^{2}Z_{i,jl}^{'}Z_{i,jl}\right]\right|Z_{i,jl},C_{i,jl}\right]\nonumber \\
 & =E\left[C_{i,jl}E\left[\left.D_{i,1}-\Lambda\left(Z_{i,jl}\theta_{jl}\right)^{2}\right|Z_{i,jl},C_{i,jl}\right]Z_{i,jl}^{'}Z_{i,jl}\right]\nonumber \\
 & =E\left[C_{i,jl}\Lambda\left(Z_{i,jl}\theta_{jl}\right)\left[1-\Lambda\left(Z_{i,jl}\theta_{jl}\right)\right]Z_{i,jl}^{'}Z_{i,jl}\right]\nonumber \\
 & \equiv\Sigma_{jl}.\label{eq:score_sjl_noD}
\end{align}
The score for the composite likelihood function reuiqres more notation.
First, note that the score in (\ref{eq:score_sjl_noD}) is for the
parameter $\theta_{jl}$, and is therefore a matrix of dimensions
at most $\left(K+1+2\right)\times\left(K+1+2\right)$ matrix. The
score contribution for the composite likelihood is necessarily a $\left(K+1+\left(J-2\right)\right)\times\left(K+1+\left(J-2\right)\right)$
matrix, with rows and columns of zeros inserted into the location
where parameters in $\theta$ are absent from $\theta_{jl}$ is called
$\Omega_{jl}$. Formally,
\begin{align*}
\tilde{s}_{i,jl}\left(\theta\right) & =\frac{\partial l_{i,jl}\left(\theta_{jl}\right)}{\partial\theta},\\
\Omega_{jl} & \equiv E\left[\tilde{s}_{i,jl}\left(\theta_{0}\right)\tilde{s}_{i,jl}\left(\theta_{0}\right)^{'}\right].
\end{align*}
The variance of the score of the composite conditional log likelihood
function is 
\begin{align}
\Omega & \equiv E\left[\left(\sum_{j,l}\tilde{s}_{i,jl}\left(\theta_{0}\right)\right)\left(\sum_{j,l}\tilde{s}_{i,jl}\left(\theta_{0}\right)\right)^{'}\right]\nonumber \\
 & =\sum_{jl}\Omega_{jl}+\sum_{\left(j,l\right)\neq\left(j',l'\right)}E\left[\tilde{s}_{i,jl}\left(\theta_{0}\right)\tilde{s}_{i,j'l'}\left(\theta_{0}\right)^{'}\right].\label{eq:score_CCML}
\end{align}
That the conditions for a CLT (cf. condition iii in NM94) are satisfied
then follows from the boundedness of $C$ and $\Lambda$, and Assumption
\ref{assu:XVariation}.

Furthermore, note that the Hessian of the $\left(j,l\right)$ contribution
is given by 
\[
E\left[H_{i,jl}\left(\theta_{jl}\right)\right]=-\Sigma_{jl},
\]
which follows immediately from comparing (\ref{eq:score_sjl_noD})
and (\ref{eq:Hessian_ijl}). To obtain a Hessian for the composite
likelihood, we enlarge the dimension of that Hessian by defining
\[
\tilde{H}_{i,jl}\left(\theta\right)=\frac{\partial^{2}l_{i,jl}\left(\theta_{jl}\right)}{\partial\theta\partial\theta^{'}}.
\]
It follows that $E\left[\tilde{H}_{i,jl}\left(\theta_{0}\right)\right]=-\Omega_{jl}$
and 
\begin{equation}
H=-\sum_{j,l}\Omega_{jl}.\label{eq:Hessian_CCML}
\end{equation}
Conditions (iv, v, Hessian) then follow from Assumption \ref{assu:XVariation}.
All conditions in Theorem 3.1 of NM94 hold, and Theorem \ref{thm:asymptotics-CCMLE}
therefore holds.
\end{proof}

\section{Additional empirical results\label{sec:Additional-empirical-results}}

We present a version of our main results in Table \ref{tab:DOFE_results}
that incldues the coefficients on the time-invariant variables, in
Table \ref{tab:DOFE_results-1}.
\begin{table}
\begin{centering}
{\footnotesize{}\hspace*{-1cm}}%
\begin{tabular}[b]{lccccccccc}
\hline 
 & {\footnotesize{}(a)} & {\footnotesize{}(b)} & {\footnotesize{}(c)} & {\footnotesize{}(d)} & {\footnotesize{}(e)} & {\footnotesize{}(f)} & {\footnotesize{}(g)} & {\footnotesize{}(h)} & {\footnotesize{}(i)}\tabularnewline
 & {\footnotesize{}DOLFE} & {\footnotesize{}DOLFE} & {\footnotesize{}DOLFE} & {\footnotesize{}DOLFE} & {\footnotesize{}FEOL} & {\footnotesize{}DOL} & {\footnotesize{}OL} & {\footnotesize{}FELM} & {\footnotesize{}LM}\tabularnewline
 & {\footnotesize{}$h=1$} & {\footnotesize{}$h=0.1$} & {\footnotesize{}$h=10$} & {\footnotesize{}$h=1$} &  &  &  &  & \tabularnewline
\hline 
{\footnotesize{}log(income)} & {\footnotesize{}0.049} & {\footnotesize{}-0.047} & {\footnotesize{}0.059} & {\footnotesize{}0.061} & {\footnotesize{}0.020} & {\footnotesize{}0.340} & {\footnotesize{}0.492} & {\footnotesize{}0.003} & {\footnotesize{}0.194}\tabularnewline
\multirow{1}{*}[103cm]{} & \multirow{1}{*}{{\tiny{}(0.033)}} & {\tiny{}(0.056)} & {\tiny{}(0.029)} & {\tiny{}(0.029)} & {\tiny{}(0.019)} & {\tiny{}(0.004)} & {\tiny{}(0.004)} & {\tiny{}(0.003)} & {\tiny{}(0.002)}\tabularnewline
{\footnotesize{}child} & {\footnotesize{}-0.030} & {\footnotesize{}0.006} & {\footnotesize{}-0.031} & {\footnotesize{}-0.026} & {\footnotesize{}0.021} & {\footnotesize{}0.060} & {\footnotesize{}0.089} & {\footnotesize{}0.002} & {\footnotesize{}0.033}\tabularnewline
 & {\tiny{}(0.051)} & {\tiny{}(0.069)} & {\tiny{}(0.050)} & {\tiny{}(0.049)} & {\tiny{}(0.032)} & {\tiny{}(0.005)} & {\tiny{}(0.005)} & {\tiny{}(0.005)} & {\tiny{}(0.002)}\tabularnewline
{\footnotesize{}married} & {\footnotesize{}0.139} & {\footnotesize{}-0.041} & {\footnotesize{}0.157} & {\footnotesize{}0.130} & {\footnotesize{}0.164} & {\footnotesize{}0.073} & {\footnotesize{}0.141} & {\footnotesize{}0.029} & {\footnotesize{}0.062}\tabularnewline
 & {\tiny{}(0.087)} & {\tiny{}(0.119)} & {\tiny{}(0.086)} & {\tiny{}(0.088)} & {\tiny{}(0.053)} & {\tiny{}(0.007)} & {\tiny{}(0.008)} & {\tiny{}(0.009)} & {\tiny{}(0.003)}\tabularnewline
{\footnotesize{}unemp} & {\footnotesize{}-0.188} & {\footnotesize{}-0.230} & {\footnotesize{}-0.178} &  & {\footnotesize{}-0.196} & {\footnotesize{}-0.242} & {\footnotesize{}-0.308} & {\footnotesize{}-0.033} & {\footnotesize{}-0.127}\tabularnewline
 & {\tiny{}(0.070)} & {\tiny{}(0.110)} & {\tiny{}(0.068)} &  & {\tiny{}(0.038)} & {\tiny{}(0.014)} & {\tiny{}(0.015)} & {\tiny{}(0.007)} & {\tiny{}(0.006)}\tabularnewline
{\footnotesize{}retired} & {\footnotesize{}-0.132} & {\footnotesize{}-0.043} & {\footnotesize{}-0.139} &  & {\footnotesize{}-0.154} & {\footnotesize{}-0.050} & {\footnotesize{}-0.097} & {\footnotesize{}-0.027} & {\footnotesize{}-0.047}\tabularnewline
 & {\tiny{}(0.082)} & {\tiny{}(0.119)} & {\tiny{}(0.080)} &  & {\tiny{}(0.041)} & {\tiny{}(0.010)} & {\tiny{}(0.011)} & {\tiny{}(0.007)} & {\tiny{}(0.004)}\tabularnewline
{\footnotesize{}other} & {\footnotesize{}-0.370} & {\footnotesize{}-0.207} & {\footnotesize{}-0.369} &  & {\footnotesize{}-0.473} & {\footnotesize{}-0.771} & {\footnotesize{}-1.087} & {\footnotesize{}-0.082} & {\footnotesize{}-0.460}\tabularnewline
 & {\tiny{}(0.061)} & {\tiny{}(0.087)} & {\tiny{}(0.061)} &  & {\tiny{}(0.040)} & {\tiny{}(0.010)} & {\tiny{}(0.012)} & {\tiny{}(0.007)} & {\tiny{}(0.005)}\tabularnewline
\hline 
{\footnotesize{}$\rho$} & {\footnotesize{}0.733} & {\footnotesize{}0.723} & {\footnotesize{}0.733} & {\footnotesize{}0.734} &  & {\footnotesize{}1.987} &  &  & \tabularnewline
 & {\tiny{}(0.020)} & {\tiny{}(0.025)} & {\tiny{}(0.020)} & {\tiny{}(0.017)} &  & {\tiny{}(0.023)} &  &  & \tabularnewline
{\footnotesize{}$\gamma_{2}$} & {\footnotesize{}-3.275} & {\footnotesize{}-3.260} & {\footnotesize{}-3.272} & {\footnotesize{}-3.211} & {\footnotesize{}-3.487} & {\footnotesize{}-2.506} & {\footnotesize{}-1.992} &  & \tabularnewline
 & {\tiny{}(0.054)} & {\tiny{}(0.068)} & {\tiny{}(0.053)} & {\tiny{}(0.048)} & {\tiny{}(0.015)} & {\tiny{}(0.007)} & {\tiny{}(0.006)} &  & \tabularnewline
{\footnotesize{}$\gamma_{4}$} & {\footnotesize{}3.326} & {\footnotesize{}3.356} & {\footnotesize{}3.329} & {\footnotesize{}3.321} & {\footnotesize{}3.997} & {\footnotesize{}3.089} & {\footnotesize{}2.603} &  & \tabularnewline
 & {\tiny{}(0.055)} & {\tiny{}(0.076)} & {\tiny{}(0.055)} & {\tiny{}(0.054)} & {\tiny{}(0.024)} & {\tiny{}(0.006)} & {\tiny{}(0.006)} &  & \tabularnewline
\hline 
{\footnotesize{}male} &  &  &  &  &  & {\footnotesize{}0.126} & {\footnotesize{}0.180} &  & {\footnotesize{}0.070}\tabularnewline
 &  &  &  &  &  & {\tiny{}(0.006)} & {\tiny{}(0.007)} &  & {\tiny{}(0.003)}\tabularnewline
{\footnotesize{}urban\_mid} &  &  &  &  &  & {\footnotesize{}-0.005} & {\footnotesize{}-0.018} &  & {\footnotesize{}-0.006}\tabularnewline
 &  &  &  &  &  & {\tiny{}(0.008)} & {\tiny{}(0.009)} &  & {\tiny{}(0.004)}\tabularnewline
{\footnotesize{}urban\_low} &  &  &  &  &  & {\footnotesize{}0.048} & {\footnotesize{}0.021} &  & {\footnotesize{}0.006}\tabularnewline
 &  &  &  &  &  & {\tiny{}(0.007)} & {\tiny{}(0.008)} &  & {\tiny{}(0.003)}\tabularnewline
{\footnotesize{}age $]25;35]$} &  &  &  &  &  & {\footnotesize{}-0.613} & {\footnotesize{}-0.772} &  & {\footnotesize{}-0.282}\tabularnewline
 &  &  &  &  &  & {\tiny{}(0.015)} & {\tiny{}(0.016)} &  & {\tiny{}(0.005)}\tabularnewline
{\footnotesize{}age $]35;45]$} &  &  &  &  &  & {\footnotesize{}-1.047} & {\footnotesize{}-1.349} &  & {\footnotesize{}-0.496}\tabularnewline
 &  &  &  &  &  & {\tiny{}(0.015)} & {\tiny{}(0.016)} &  & {\tiny{}(0.005)}\tabularnewline
{\footnotesize{}age $]45;55]$} &  &  &  &  &  & {\footnotesize{}-1.423} & {\footnotesize{}-1.938} &  & {\footnotesize{}-0.722}\tabularnewline
 &  &  &  &  &  & {\tiny{}(0.015)} & {\tiny{}(0.015)} &  & {\tiny{}(0.005)}\tabularnewline
{\footnotesize{}age $]55;65]$} &  &  &  &  &  & {\footnotesize{}-1.392} & {\footnotesize{}-2.016} &  & {\footnotesize{}-0.747}\tabularnewline
 &  &  &  &  &  & {\tiny{}(0.016)} & {\tiny{}(0.017)} &  & {\tiny{}(0.006)}\tabularnewline
{\footnotesize{}age $]65;\infty[$} &  &  &  &  &  & {\footnotesize{}-1.549} & {\footnotesize{}-2.278} &  & {\footnotesize{}-0.854}\tabularnewline
 &  &  &  &  &  & {\tiny{}(0.017)} & {\tiny{}(0.019)} &  & {\tiny{}(0.007)}\tabularnewline
{\footnotesize{}Primary schooling} &  &  &  &  &  & {\footnotesize{}0.385} & {\footnotesize{}0.510} &  & {\footnotesize{}0.204}\tabularnewline
 &  &  &  &  &  & {\tiny{}(0.028)} & {\tiny{}(0.031)} &  & {\tiny{}(0.013)}\tabularnewline
{\footnotesize{}Lower secondary} &  &  &  &  &  & {\footnotesize{}0.516} & {\footnotesize{}0.737} &  & {\footnotesize{}0.302}\tabularnewline
 &  &  &  &  &  & {\tiny{}(0.029)} & {\tiny{}(0.031)} &  & {\tiny{}(0.013)}\tabularnewline
{\footnotesize{}Upper secondary} &  &  &  &  &  & {\footnotesize{}0.694} & {\footnotesize{}0.952} &  & {\footnotesize{}0.384}\tabularnewline
 &  &  &  &  &  & {\tiny{}(0.028)} & {\tiny{}(0.031)} &  & {\tiny{}(0.012)}\tabularnewline
{\footnotesize{}Post-secondary} &  &  &  &  &  & {\footnotesize{}0.717} & {\footnotesize{}0.931} &  & {\footnotesize{}0.377}\tabularnewline
 &  &  &  &  &  & {\tiny{}(0.031)} & {\tiny{}(0.034)} &  & {\tiny{}(0.014)}\tabularnewline
{\footnotesize{}Tertiary} &  &  &  &  &  & {\footnotesize{}0.890} & {\footnotesize{}1.233} &  & {\footnotesize{}0.490}\tabularnewline
 &  &  &  &  &  & {\tiny{}(0.029)} & {\tiny{}(0.031)} &  & {\tiny{}(0.013)}\tabularnewline
\hline 
\end{tabular}{\footnotesize\par}
\par\end{centering}
\caption{Main results.}

\centering{}\label{tab:DOFE_results-1}
\end{table}
 In the main text, we established that controlling for unobserved
heterogeneity is important, so we should be careful in interpreting
results from models without fixed effects (columns f, g, i). We find
that men demonstrate higher levels of self-reported health (Caroli
and Weber-Baghdiguian, 2016; Bago d\textquoteright Uva et al., 2008),
and those living in rural areas are also more likely to report better
health (Lindeboom and van Doorslaer, 2004). Results also show that
age is negatively related to self-reported health (Bago d\textquoteright Uva
et al., 2008; Lindeboom and van Doorslaer, 2004); and that individuals
with higher levels of education report better health (Conti et al.,
2010).
\end{document}